\documentclass{jcmlatex}
\usepackage{epstopdf}
\usepackage{siunitx}
\usepackage{float}
\usepackage{rotating} 
\usepackage{algorithm} 
\usepackage{amsthm}
\usepackage{lscape}
\usepackage{amsmath}    
\usepackage{amssymb}       
\usepackage{algpseudocode}
\usepackage{booktabs}
\usepackage{multirow} 
\usepackage{subcaption} 
\usepackage[colorlinks=true, linkcolor=blue, citecolor=blue, urlcolor=blue]{hyperref}
\usepackage[numbers, sort&compress]{natbib} 

\def\JCMvol{xx}
\def\JCMno{x}
\def\JCMyear{200x}
\def\JCMreceived{xxx}
\def\JCMrevised{xxx}
\def\JCMaccepted{xxx}
\def\JCMonline{}
\begin{document}

\markboth{}{}

\title{An improved boundary-focused adaptive quadtree algorithm for circle-polygon intersection area approximation}

\author{Zeping Yi
\thanks{School of Mathematical Sciences, Beihang University,  Beijing, China \\ Email: yzping@buaa.edu.cn}
\and
Lan Li
\thanks{School of Mathematical Sciences, Beihang University, Beijing, China\\ Email: li22377185@buaa.edu.cn}}

\maketitle

	\begin{abstract}

In this paper, we present an improved boundary-focused adaptive quadtree algorithm for the intersection area of multiple circles and a complex polygon. Based on the classical quadtree and boundary-layer error analysis, we propose two variants with different splitting and leaf-evaluation schemes. The deterministic variant uses exact predicates to classify cells as interior, exterior, or boundary. Interior and exterior cells are evaluated exactly, boundary cells are subdivided to a prescribed resolution, and accepted boundary leaves are evaluated by a midpoint rule on an $s\times s$ subgrid. The error-budget variant replaces the fixed resolution with a statistical per-cell error budget controlled by sample size $N_{\min}$ and tolerance constant $C$. Each visited cell is probed with $N_{\min}$ uniform points, and a statistical criterion decides acceptance or splitting. We derive a boundary-layer error bound and prove that a prescribed relative error $\varepsilon_{\mathrm{rel}}$ is achieved with $O((m+n)/\varepsilon_{\mathrm{rel}})$ operations, with $m,n$ the numbers of polygon vertices and circles, improving on the $O(1/\varepsilon_{\mathrm{rel}}^{2})$ cells of uniform grids and Monte Carlo samples. Experiments on synthetic and real instances achieve median relative errors of $10^{-6}$, outperforming five classical estimators as confirmed by non-parametric significance tests. Sensitivity analysis shows the error-budget variant is robust over a wide parameter range.

\end{abstract}



\section{Introduction}
Computation of effective coverage area resulting from the intersection of multiple circles and a complex polygon is a fundamental geometric problem with broad applications across various scientific and engineering disciplines \cite{ROCHA2019173,Shaikh2025,Wu2025,Bogosel2026}. In wireless sensor networks, it is essential for evaluating the coverage quality of a sensor node (modeled as a circle) over an irregular target region (modeled as a polygon) \cite{Shaikh2025}. In robotics, it arises in path planning for cleaning or inspection tasks, where a robot with a circular tool head must cover a polygonal workspace \cite{Wu2025}. Similar requirements also arise in many industrial applications \cite{LI2025111431,LAI2025107099}.

Given its importance, numerous methods have been developed to compute the circle-polygon intersection area. Boundary integration \cite{Hall1994} and closed-form overlap formulas for disks and ellipses \cite{KIM2025425,COLL201959} are representative analytical methods for exact solutions. It typically converts the area of  intersection region into a boundary curve integral via Green's theorem to achieve high accuracy. However, its computational cost increases when dealing with intersections in complex regions. Another category consists of semi-analytical methods, with triangulation as a typical example, which  triangulates the polygon and then computes the area of each triangle-circle intersection analytically \cite{sead2018}. Triangulation could obtain exact solutions at the expense of high computational cost for simple domains. Nevertheless, both accuracy and efficiency decline for complex geometries with multiple circles. To overcome these challenges, numerical methods such as Monte Carlo \cite{2024,Altmann02012026}, uniform grid \cite{AKMAN1989410,Bamberger,SATO2019657}, and grid integration \cite{CHEN2026122091,KEQIN2026107852} have been developed. Unlike analytical methods, these numerical algorithms avoid the need for robust handling of numerous special cases arising from edge-circle intersections. However, their accuracy heavily depends on the number of polygon vertices, the complexity of the polygonal shape, and the number of circles. In particular, the Monte Carlo method suffers from slower convergence and larger error in complex scenarios \cite{Pellegrini1997,ALANAZI2026102146} and uniform grid method expends unnecessary computational effort on geometrically simple interior regions. Therefore, researchers improved the uniform grid method and proposed an adaptive subdivision approach \cite{Zhao2026,YIN2026104102}, which maintains high accuracy in complex scenarios but suffers from longer running time. 

These observations suggest that both accuracy and computational effort should be concentrated where the boundary of the covered region actually lies. Therefore, we propose a \emph{boundary-focused adaptive quadtree} algorithm. Our contributions are summarized as follows:
\begin{itemize}
	\item \textbf{Algorithm.} We develop an adaptive quadtree with exact cell classification for the intersection of a complex polygon and a union of circles. The refinement is concentrated on the boundary, while interior and exterior cells are evaluated exactly. For boundary cells, two leaf-evaluation variants are provided, a deterministic midpoint rule and an error-budget rule based on uniform per-cell samples and the Wilson score interval.
	
	\item \textbf{Error and complexity analysis.} For the deterministic variant, we derive the boundary-layer error  bound \cite{4767448}. Based on the bound, under the standard scaling $A=\Theta(L^{2})$ and $\operatorname{Perim}(\Gamma)=O(L)$ (with $A$  the area of the covered region, $L$ the largest side length of the bounding box and $\Gamma$ the boundary of the covered region), a target relative error $\varepsilon_{\mathrm{rel}}$ is achieved in $O((m+n)/\varepsilon_{\mathrm{rel}})$ operations. This compares with $O(1/\varepsilon_{\mathrm{rel}}^{2})$ cells for a uniform grid and $O(1/\varepsilon_{\mathrm{rel}}^{2})$ Monte Carlo samples in the root-mean-square sense. For the error-budget variant, a per-cell probability bound follows directly from the Wilson score interval.
	
	\item \textbf{Numerical experiments.} On two controlled synthetic studies, the proposed algorithm attains a median relative error at the $10^{-6}$ level, with all values within $10^{-7}$--$10^{-5}$, and is one to two orders of magnitude more accurate than Monte Carlo, uniform grid, adaptive subdivision, grid integration, and triangulation overall. Friedman tests, Holm-corrected Wilcoxon signed-rank tests, and Cliff's $\delta$ confirm that the advantage is statistically robust. On the real Caribbean coastline instance with $1{,}284$ vertices, $97$ polygon parts, and $100$ circles, the method reaches a relative error of $0.0001\%$ in $1.15$ s, the smallest error among all approximate methods. Sensitivity analysis further shows that the accuracy of the error-budget variant is governed mainly by its tolerance and depends only weakly on the per-cell sample size $N_{min}$.
\end{itemize}

The remainder of this paper is organized as follows. In Section 2, we formally define the problem and explain necessary theorems. Then we detail the proposed adaptive quadtree algorithm with boundary focusing in Section 3. Numerical experiments and comparisons are presented in Section 4. Finally, we conclude the paper and discuss future work in Section 5.

 \section{Preliminaries}
\subsection{Green's theorem}
\begin{definition}
Given a set of circles $\left\{ C_i \right\} _{i=1}^{n}$ and a polygon $P$  defined by a counterclockwise sequence of vertices $\left\{\boldsymbol{v}_i \right\} _{i=1}^{m}$. The problem is to compute the area of their intersection $P \cap \left( \bigcup_{i=1}^{n} C_i \right)$. The intersection region is $\Omega = P \cap \left(\bigcup_{i=1}^{n} C_i \right)$, as shown in Figs.~\ref{fig12} and~\ref{fig13}. Let $B$ be the axis-aligned bounding box \cite{AHMAD2026112710,OZTUNER2026104798} containing $P$ and $\left\{ C_i \right\} _{i=1}^{n}$ completely.
	\label{df1}
\end{definition}

 \begin{figure}[H]
	\centering
	\includegraphics[width=0.26\textwidth]{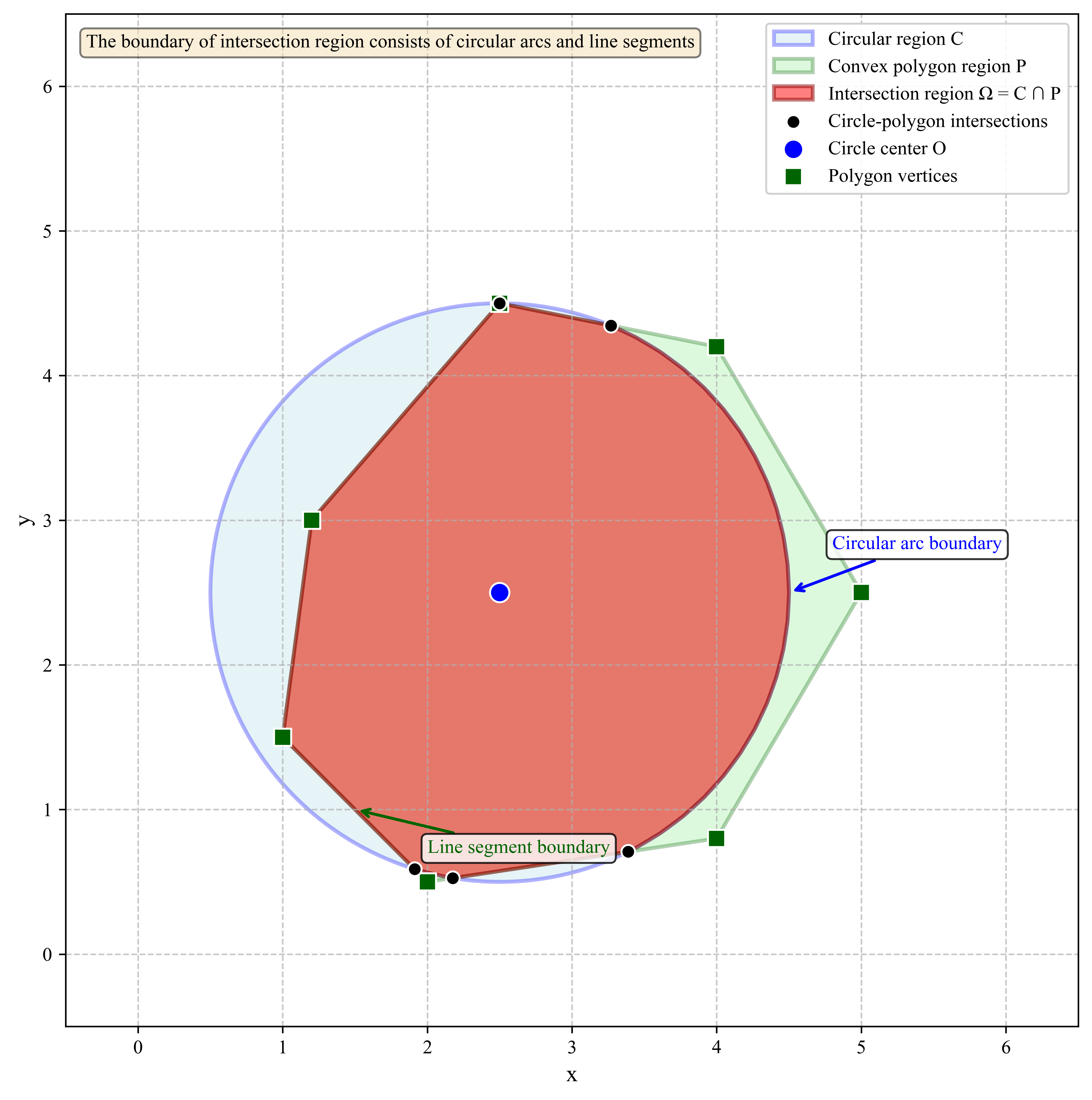}
	\caption{Intersection of a single circle and a polygon}
	\label{fig12}
\end{figure}

\begin{figure}[H]
	\centering
	\includegraphics[width=0.55\textwidth]{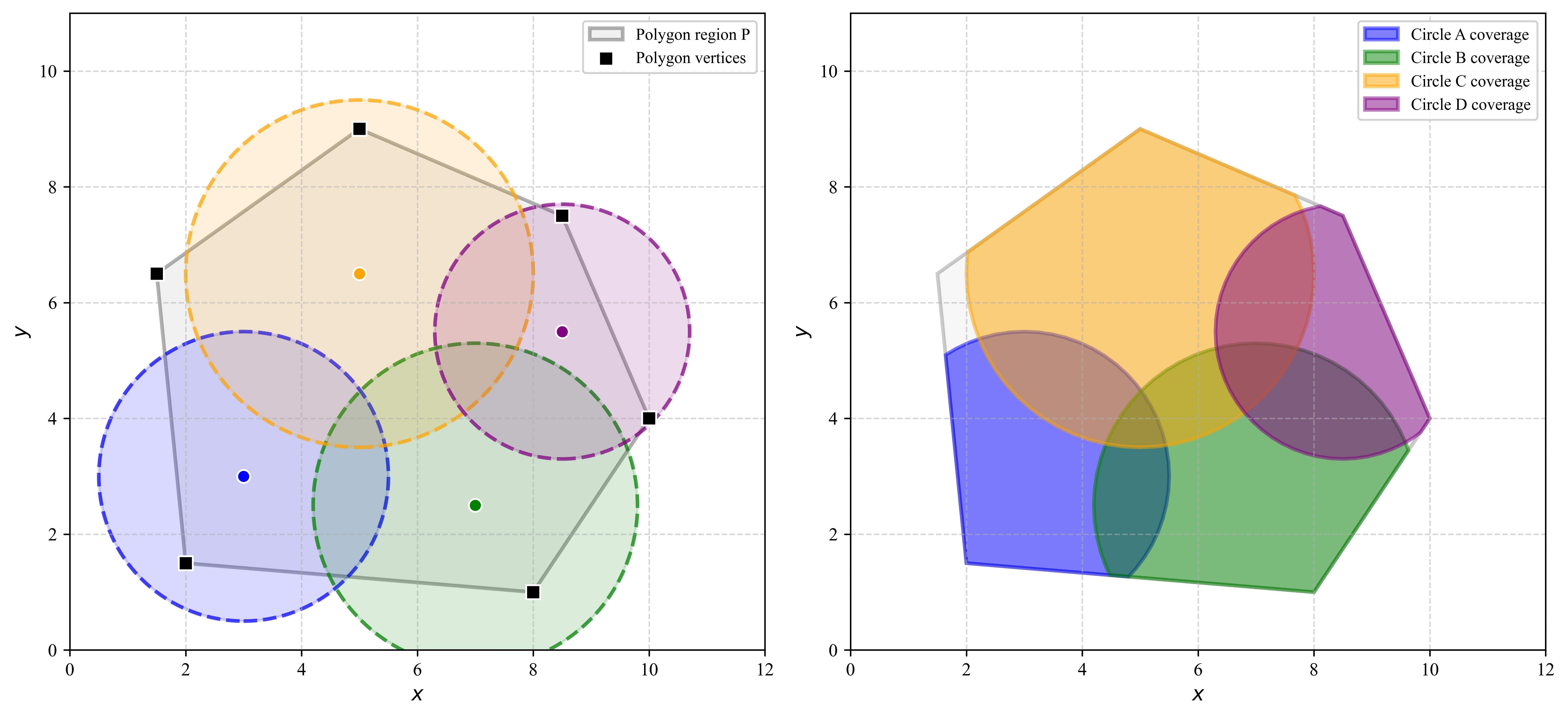}
	\caption{Intersection of multiple circles with a polygon.}
	\label{fig13}
\end{figure}

\begin{theorem}
	Let $\Omega \subset \mathbb{R}^2$ be a bounded region whose boundary $\partial \Omega$ is a piecewise smooth, simple closed curve. Then the area of $\Omega$ can be expressed as $S = \iint_{\Omega} dx\,dy = \frac{1}{2} \oint_{\partial \Omega} (x\,dy - y\,dx)$ according to Green's Theorem, where the line integral is traversed in the counterclockwise orientation.
\end{theorem}

\begin{proof}
	Green's Theorem states that for any continuously differentiable functions $P(x,y)$ and $Q(x,y)$ defined on a region $\Omega$ with piecewise smooth boundary $\partial \Omega$,
	\begin{equation}
		\iint_{\Omega} \left( \frac{\partial Q}{\partial x} - \frac{\partial P}{\partial y} \right) dx\,dy = \oint_{\partial \Omega} (P\,dx + Q\,dy).
	\end{equation}
	Choose $P(x,y) = -y/2$ and $Q(x,y) = x/2$. Then
	$	\frac{\partial Q}{\partial x} - \frac{\partial P}{\partial y} = 1.$
	Substituting into Green's Theorem yields
	\begin{equation}
		S = \iint_{\Omega} 1\,dx\,dy = \frac{1}{2} \oint_{\partial \Omega} (-y\,dx + x\,dy) = \frac{1}{2} \oint_{\partial \Omega} (x\,dy - y\,dx).
	\end{equation}
\end{proof}

\begin{remark}
	If the boundary $\partial \Omega$ consists of $m$ piecewise smooth curve segments (circular arcs or line segments) $\Gamma_1, \Gamma_2, \ldots, \Gamma_m$ joined end-to-end, then the area can be expressed as the sum of integrals over each segment:
	\begin{equation}
		\label{eq:1}
		S = \frac{1}{2} \sum_{j=1}^{m} \int_{\Gamma_j} (x\,dy - y\,dx).
	\end{equation}
	The orientation of $\partial \Omega$ is still counterclockwise. However, the practical implementation of this exact approach requires identifying all boundary segments $\Gamma_j$ and their endpoints robustly, handling a large number of small segments when the polygon has many vertices and managing numerous special cases where vertices lie exactly on the circle or edges are tangent. These difficulties motivate replacing the exact boundary integral by a boundary-focused adaptive sampling scheme. Instead of integrating along the closed curve $\partial\Omega$, the region can be partitioned and only cells adjacent to the boundary are refined and sampled.
\end{remark}

 \subsection{Monte Carlo method}
The Monte Carlo method is a classical numerical approach for approximating area when analytical solutions are difficult to obtain. For computing circle-polygon intersection area, the basic idea is to uniformly sample random points within a bounding domain and count the proportion of points that fall inside the target region \cite{Wu2026}. According to Definition \ref{df1}, suppose $N$ independent and uniformly distributed random points $\{\boldsymbol{p}_i\}_{i=1}^{N}$ are sampled from $B$.

\begin{proposition}
	Define the indicator function
	\begin{equation}
		\mathbf{1}_{\Omega}(\boldsymbol{p}) = 
		\begin{cases}
			1, & \text{if } \boldsymbol{p} \in \Omega,\\
			0, & \text{otherwise}. 
		\end{cases}.
	\end{equation}
	Then the area of $\Omega$ can be estimated by
	\begin{equation}
		\hat{A}= A_{\mathrm{bbox}} \cdot \frac{1}{N} \sum_{i=1}^{N} \mathbf{1}_{\Omega}(\boldsymbol{p}_i).
	\end{equation}
\end{proposition}

\begin{theorem}
	\label{thm:3}
Let $A$  be the true area of the region $\Omega$, and let $A_{\text{bbox}}$ be the area of the bounding box $B$. Define $\hat{A} = A_{\text{bbox}} \cdot N_{\text{in}} / N$, where $N$ is the total number of uniform random samples in the bounding box and $N_{\text{in}}$ is the number of samples falling in $\Omega$. Then for a confidence level $1-\alpha$:
	
	\begin{enumerate}
		\item The estimator $\hat{A}$ is unbiased: $\mathbb{E}[\hat{A}] = A$.
		
		\item The variance satisfies $\displaystyle \text{Var}(\hat{A}) \le \dfrac{A_{\text{bbox}}^2}{4N}$.
		
		\item For large $N$, with probability approximately $1-\alpha$,
	$|\hat{A}-A|\le z_{1-\alpha/2}A_{\mathrm{bbox}}/(2\sqrt{N})=O(N^{-1/2})$, where $z_{1-\alpha/2}$ is the $(1-\alpha/2)$-quantile of the standard normal distribution.
		
		\item To guarantee an absolute error $\eta$ with probability $1-\alpha$, $N=\Theta(A_{\mathrm{bbox}}^{2}/\eta^{2})$ samples suffice, i.e. $N=O(1/\eta^{2})$ for a fixed bounding box.
		
		\item The total computational complexity is $T = O((n + m)/\eta^2)$, where $m$ is the number of polygon vertices and $n$ is the number of circles.
	\end{enumerate}
\end{theorem}

\begin{proof}
	\textbf{(1) Unbiasedness.} 
	Let $p = A / A_{\text{bbox}}$ be the true coverage ratio. Since $N_{\text{in}} \sim \text{Binomial}(N, p)$, 
	we have $\mathbb{E}[N_{\text{in}}] = Np$. Therefore,
	\begin{equation}
		\mathbb{E}[\hat{A}] = A_{\text{bbox}} \cdot \frac{\mathbb{E}[N_{\text{in}}]}{N} = A_{\text{bbox}} \cdot p = A.
	\end{equation}
	
	\textbf{(2) Variance bound.} 
	Since $\text{Var}(N_{\text{in}}) = Np(1-p)$,
	\begin{equation}
		\text{Var}(\hat{A}) = \frac{A_{\text{bbox}}^2}{N^2} \cdot Np(1-p) = \frac{A_{\text{bbox}}^2 p(1-p)}{N}.
	\end{equation}
	
	For any $p \in [0,1]$, $p(1-p) \le 1/4$. Therefore,
	\begin{equation}
		\text{Var}(\hat{A}) \le \frac{A_{\text{bbox}}^2}{4N}.
	\end{equation}
	
\textbf{(3) Error bound.}
For a large $N$, the central limit theorem gives $\hat{A}\sim\mathcal{N}(A,\operatorname{Var}(\hat{A}))$.
Thus, with probability approximately $1-\alpha$,
\begin{equation}
P\left(|\hat{A}-A|\le z_{1-\alpha/2}\sqrt{\operatorname{Var}(\hat{A})}\right)\approx 1-\alpha,
\end{equation}
and substituting the variance bound yields,
\begin{equation}
	|\hat{A} - A| \le z_{1-\alpha/2} \cdot \frac{A_{\mathrm{bbox}}}{2\sqrt{N}} = O\left(\frac{1}{\sqrt{N}}\right).
\end{equation}

\textbf{(4) Sample size requirement.}
To guarantee $|\hat{A}-A|\le\eta$ with probability $1-\alpha$, it suffices to have
\begin{equation}
	\frac{z_{1-\alpha/2} A_{\mathrm{bbox}}}{2\sqrt{N}} \le \eta
	\quad\Longrightarrow\quad
	N \ge \frac{z_{1-\alpha/2}^{2} A_{\mathrm{bbox}}^{2}}{4\eta^{2}}.
\end{equation}
Thus $N = O(1/\eta^2)$ for a fixed bounding box.
	
	\textbf{(5) Complexity.} 
	Each sample requires $O(m)$ operations \cite{DAI2026103994} for point-in-polygon test and $O(n)$ operations \cite{Hilden1972}
	for point-in-circle tests, totaling $O(n + m)$ per sample. Therefore,
	\begin{equation}
		T = N \cdot O(n + m) = O\left(\frac{n + m}{\eta^2}\right).
	\end{equation}
	For fixed $n$ and $m$, this simplifies to $T = O(1/\eta^2)$. 
\end{proof}

\subsection{Quadtree}
A quadtree is a hierarchical spatial data structure that recursively partitions a bounded domain into four equal-sized quadrants. It is widely used in computational geometry \cite{OHEIM2026118839}, computer graphics \cite{TONG2026115820}, and pattern recognition \cite{Ni2026}  due to its ability to adaptively refine regions of interest while maintaining coarse resolution elsewhere.

\begin{definition}[Quadtree node]
	A quadtree node $E$ represents an axis-aligned rectangular cell
\begin{equation}
	E=[x_{\min},x_{\max}]\times[y_{\min},y_{\max}],
\end{equation}
with width $w(E)=x_{\max}-x_{\min}$ and height $h(E)=y_{\max}-y_{\min}$. The cell is square when $w(E)=h(E)$.
\end{definition}

\begin{definition}[Quadtree partitioning]
Given a root cell $E_0$, the quadtree recursively subdivides a cell $E$ into four children of size $\tfrac{w(E)}{2}\times\tfrac{h(E)}{2}$. Subdivision continues until a termination condition is satisfied.
\end{definition}

\begin{definition}[Leaf and internal nodes]
	A node that has no children is called a leaf node. A node that has four children is called an internal node. Each node has either four children (subdivided) or zero children (leaf node). The set of leaf nodes forms a partition of the original domain. 
\end{definition}

\begin{proposition}[Quadtree node count]
	\label{lem:quadtree_nodes}
	For a quadtree with depth $D$ (i.e., the root is at depth $0$ and the deepest leaf node at depth $D$), the maximum number of leaf nodes is $N_{\text{leaf}} = 4^D$.
\end{proposition}

\begin{proof}
	At depth $d$, there are at most $4^d$ nodes. The total number of nodes in a full quadtree is
	\[
	\sum_{d=0}^{D} 4^d = \frac{4^{D+1}-1}{3} = O(4^D) = O(N_{\text{leaf}}).
	\]
	Thus, internal nodes and leaf nodes are asymptotically of the same order.
\end{proof}

\begin{theorem}[Quadtree maximum depth]
	\label{lem:depth_size}
	Let $L$ be the largest side length of the root cell $E_0$ (the bounding box) and suppose its aspect ratio is bounded by a constant, so that both sides are $\Theta(L)$. Let $\varepsilon_{\mathrm{rel}}\in(0,1]$ be the prescribed relative error and set $\delta=\varepsilon_{\mathrm{rel}}L$. Then
\begin{equation}
	D_{\max}=\left\lceil\log_2\frac{L}{\delta}\right\rceil
	=\left\lceil\log_2\frac{1}{\varepsilon_{\mathrm{rel}}}\right\rceil
	=O\bigl(\log_2(1/\varepsilon_{\mathrm{rel}})\bigr).
\end{equation}
\end{theorem}

\begin{proof}
	Let $W$ and $H$ be the width and height of the root cell $E_0$, and set $L=\max\{W,H\}$. By the bounded-aspect-ratio assumption, $W=\Theta(H)$, hence $L=\Theta(W)=\Theta(H)$. After $d$ subdivisions both sides are halved, so the sides are $W/2^{d}$ and $H/2^{d}$, and the larger one is $L/2^{d}$. Refinement stops when this larger side no longer exceeds the threshold $\delta$, i.e.,
	\begin{equation}
		\frac{L}{2^{d}}\le \delta
		\ \Longrightarrow\ 2^{d}\ge \frac{L}{\delta}
		\ \Longrightarrow\ d\ge \log_2\!\left(\frac{L}{\delta}\right).
	\end{equation}
	Therefore, the maximum depth is the smallest integer satisfying this inequality,
	$D_{\max}=\left\lceil \log_2 (L/\delta)\right\rceil$.
	Substituting $\delta=\varepsilon_{\mathrm{rel}}L$ gives $L/\delta=1/\varepsilon_{\mathrm{rel}}$, and hence, for $\varepsilon_{\mathrm{rel}}\in(0,1]$,
	\begin{equation}
		D_{\max}=\left\lceil \log_2\!\left(\frac{1}{\varepsilon_{\mathrm{rel}}}\right)\right\rceil
		\le \log_2\!\left(\frac{1}{\varepsilon_{\mathrm{rel}}}\right)+1
		=O\bigl(\log_2(1/\varepsilon_{\mathrm{rel}})\bigr).
	\end{equation}
With this choice the relative resolution is $\delta/L=\varepsilon_{\mathrm{rel}}$, consistent with the relative-error bound of Proposition~\ref{prop:approx}.
\end{proof}

\begin{remark}
	The key advantage of quadtree over uniform grids is its adaptivity. Instead of refining all cells uniformly, only cells that satisfy a refinement criterion (e.g., those intersecting boundaries) are subdivided. This concentrates computational resources where they are most needed, leading to significant savings in computational time.
\end{remark}

\section{Boundary-focused adaptive quadtree algorithm}
\subsection{Adaptive spatial partitioning}

Let $B=[x_{\min},x_{\max}]\times[y_{\min},y_{\max}]$ be the axis-aligned bounding box of the polygon $P$ and the set of circles $\left\{ C_i \right\} _{i=1}^{n}$. Starting from the root cell $E_0=B$, we recursively subdivide a cell into four children of half the width and half the height, and classify every cell with respect to the covered region
\begin{equation}
	\Omega=P\cap\Bigl(\bigcup_{k=1}^{n}C_k\Bigr),
	\qquad
	\Gamma=\partial P\cup\partial\Bigl(\bigcup_{k=1}^{n}C_k\Bigr).
	\label{eq:OmegaGamma}
\end{equation}
Cells near $\Gamma$ are refined to the finest level, whereas cells far from $\Gamma$ are accepted at a coarse level. This adaptivity is the source of the savings with respect to a uniform grid. The two variants share this partitioning but differ in the cell test and in the leaf evaluation. The deterministic variant classifies each cell exactly and subdivides boundary cells up to a depth limit (Algorithm~\ref{alg:det}), whereas the error-budget variant samples each visited cell and stops refinement through a statistical error budget (Algorithm~\ref{alg:budget}).

\begin{definition}[Cell type]
	\label{def:celltype}
	Let $E\subset\mathbb{R}^{2}$ be a partitioned cell, and let $\Omega$ and $\Gamma$ be given by Eq.\eqref{eq:OmegaGamma}.
	\begin{enumerate}
		\item \textbf{Interior cell (INT)}: $E\cap\Gamma=\varnothing$ and $E$ is contained in $P$ and in at least one circle $C_k$;
		\item \textbf{Exterior cell (EXT)}: $E\cap\Omega=\varnothing$; this holds in particular if $E\cap P=\varnothing$ or if $E\cap\bigcup_{k=1}^{n}C_k=\varnothing$;
		\item \textbf{Boundary cell (BDY)}: $E$ is neither INT nor EXT.
	\end{enumerate}
\end{definition}

An INT cell lies entirely in $\Omega$ and contributes exactly $\mathrm{Area}(E)$; an EXT cell contributes nothing. The  condition 1 in Definition~\ref{def:celltype} is sufficient rather than necessary: a cell covered only by the union of two or more overlapping circles, but by no single circle, is not marked INT and is refined further. This does not affect correctness and it only adds a few subdivisions.

\begin{definition}[Diameter of a polygon]
	\label{def:diameter}
	The diameter of a polygon $P$ is
	\begin{equation}
		\mathrm{diam}(P)=\max_{u,v\in P}\lVert u-v\rVert.
	\end{equation}
\end{definition}

\begin{remark}[\textbf{Realization of the classification}]
	\label{rem:classify}
	In the deterministic variant the three cell types are evaluated with three elementary predicates, applied to a whole level of the tree at once. (i) The four corners of $E$ are tested against $P$ by ray casting. (ii) Every polygon edge is tested for intersection with the closed rectangle $E$. if no edge meets $E$ and all four corners lie in $P$, then $E\subset P$, and if no edge meets $E$ and no corner lies in $P$, then $E\cap P=\varnothing$. (iii) For each circle, let $d_k$ be the distance from the centre of $E$ to the centre of $C_k$ and let $\rho(E)=\tfrac12\sqrt{w(E)^{2}+h(E)^{2}}$ be the circumradius of $E$. Then $E\subseteq C_k$ whenever $d_k+\rho(E)\le R_k$, and $E\cap C_k=\varnothing$ whenever $d_k-\rho(E)\ge R_k$. Cells for which neither an inclusion nor an exclusion test succeeds are labelled BDY.
\end{remark}

\begin{remark}[\textbf{Termination and node counts}]\label{rem:depth}
	A cell is not subdivided if it is INT or EXT, or if it is BDY and a prescribed refinement predicate $\mathrm{Refine}(E,d)$ is false. The deterministic variant uses
	\begin{equation}
		\mathrm{Refine}_{\mathrm{det}}(E,d)\iff d<D
		\ \text{ and }\ \min\{w(E),h(E)\}\ge w_{\min},
		\label{eq:refine-det}
	\end{equation}
	so that the smallest accepted cell has side at least $\delta$ with $\delta=2^{-D}L$, where $L=\max\{x_{\max}-x_{\min},\,y_{\max}-y_{\min}\}$ is the largest side of the root cell $E_0$. Choosing $D=\lceil\log_{2}(L/\delta)\rceil=\lceil\log_{2}(1/\varepsilon_{\mathrm{rel}})\rceil$ with $\delta=\varepsilon_{\mathrm{rel}}L$, as in Theorem~\ref{lem:depth_size}, a full quadtree has at most $4^{d}$ cells at depth $d$ and hence $\sum_{d=0}^{D}4^{d}=O(4^{D})=O((L/\delta)^{2})$ nodes in total. This is a loose bound that is not attained. Since only cells meeting $\Gamma$ are subdivided, the number of BDY leaves is bounded by the number of side-$\delta$ cells needed to cover $\Gamma$,
	\begin{equation}
		N_{\mathrm{bdy}}\le \frac{\mathrm{Perim}(\Gamma)}{\delta}+O(1),
		\label{eq:ncell}
	\end{equation}
	and the total number of cells is $N_{\mathrm{cell}}=O(N_{\mathrm{bdy}})=O(\mathrm{Perim}(\Gamma)/\delta)=O(L/\delta)$ after the normalization $\mathrm{Perim}(\Gamma)=O(L)$. Here
	$\mathrm{Perim}(\Gamma)\le \mathrm{Perim}(P)+2\pi\sum_{k=1}^{n}R_k$.
\end{remark}

\begin{algorithm}[htbp]
	\caption{Deterministic boundary-focused quadtree}\label{alg:det}
	\begin{algorithmic}[1]
		\Require polygon $P$, circles $\{C_k\}_{k=1}^{n}$, depth limit $D$, width safeguard $w_{\min}$, midpoint resolution $s$
		\Ensure covered area $S$
		\State $E_0\gets$ axis-aligned bounding box of $P\cup\bigcup_{k}C_k$
		\State $S\gets 0$
		\Function{Partition}{$E,d$}
		\State classify $E$ as INT, EXT, or BDY by Definition~\ref{def:celltype}
		\If{$E$ is EXT}
		\State \Return \Comment{contributes $0$}
		\EndIf
		\If{$E$ is INT}
		\State $S\gets S+\mathrm{Area}(E)$; \Return
		\EndIf
		\If{$d=D$ \textbf{or} $\min\{w(E),h(E)\}<w_{\min}$}
		\State $S\gets S+\mathrm{Area}(E)\,\hat p_s(E)$ \Comment{midpoint rule Eq.\eqref{eq:midpoint}}
		\State \Return
		\EndIf
		\For{each of the four children $E_c$ of $E$}
		\State \Call{Partition}{$E_c,d+1$}
		\EndFor
		\EndFunction
		\State \Call{Partition}{$E_0,0$}
		\State \Return $S$
	\end{algorithmic}
\end{algorithm}

\subsection{Evaluation on boundary leaves}

Every BDY leaf is evaluated as
\begin{equation}
	\Delta S(E)=\mathrm{Area}(E)\,\hat p(E),
	\qquad
	\hat p(E)\approx
	\frac{\mathrm{Area}(E\cap\Omega)}{\mathrm{Area}(E)},
	\label{eq:leaf}
\end{equation}
where the coverage fraction $\hat p(E)$ is obtained by one of two rules that differ in how the samples are placed and in the stopping decision.

\textbf{Deterministic depth-limited rule.}
An accepted BDY leaf is evaluated by the midpoint rule on an $s\times s$ subgrid of $E$,
\begin{equation}
	\hat p_s(E)=\frac{1}{s^{2}}\sum_{j=1}^{s^{2}}\mathbf{1}_{\Omega}(x_j),
	\label{eq:midpoint}
\end{equation}
with $s=1$ (the centre point, used for the area comparisons) or $s=8$ (the smoothing used for the convergence study in Section 4). Refinement stops as in Eq.\eqref{eq:refine-det}. In the reported experiments the depth limit is reached before the width safeguard, so an accepted BDY leaf has side exactly $\delta$ and area $\delta^{2}$. Its worst-case error is estimated in Proposition~\ref{prop:approx}.

\textbf{Error-budget rule.}
Each visited cell is probed with $N_{\min}$ independent uniform points, giving a coverage fraction $\hat p(E)$. Let
\begin{equation}
	\Delta(E)=\frac{Z}{1+Z^{2}/N_{\min}}
	\sqrt{\frac{\hat p(E)\,(1-\hat p(E))}{N_{\min}}
		+\frac{Z^{2}}{4N_{\min}^{2}}},
	\qquad Z=z_{1-\alpha/2},
	\label{eq:wilson}
\end{equation}
be the half-width of the $(1-\alpha)$ Wilson score interval ($Z=1.96$ for $\alpha=0.05$). The cell is accepted as soon as $\mathrm{Area}(E)\,\Delta(E)\le\tau A_{\mathrm{bbox}}$, where $A_{\mathrm{bbox}}=(x_{\max}-x_{\min})(y_{\max}-y_{\min})$ is the bounding-box area and $\tau$ the dimensionless tolerance, and is split into four children otherwise:
\begin{equation}
	\mathrm{Refine}_{\mathrm{bud}}(E,d)\iff
	\mathrm{Area}(E)\,\Delta(E)>\tau A_{\mathrm{bbox}}
	\ \text{ and }\ d<D
	\ \text{ and }\ \min\{w(E),h(E)\}\ge w_{\min},
	\label{eq:refine-bud}
\end{equation}
where the last two conditions are safeguards (in our runs $D=18$ and $w_{\min}=10^{-4}$ degrees). The samples that trigger acceptance are reused for the contribution Eq.\eqref{eq:leaf}. No resampling is performed.

\begin{algorithm}[htbp]
	\caption{Error-budget boundary-focused quadtree}\label{alg:budget}
	\begin{algorithmic}[1]
		\Require polygon $P$, circles $\{C_k\}_{k=1}^{n}$, tolerance $\tau$, sample size $N_{\min}$, safeguards $(D,w_{\min})$, confidence level $1-\alpha$
		\Ensure covered area $S$
		\State $E_0\gets$ axis-aligned bounding box of $P\cup\bigcup_{k}C_k$
		\State $S\gets 0$
		\Function{Budget}{$E,d$}
		\State draw $N_{\min}$ i.i.d.\ uniform points in $E$ and count those in $\Omega$
		\State $\hat p(E)\gets$ fraction of sampled points in $\Omega$
		\State compute $\Delta(E)$ by \eqref{eq:wilson}
		\If{$\mathrm{Area}(E)\,\Delta(E)\le\tau A_{\mathrm{bbox}}$ \textbf{or} $d=D$ \textbf{or} $\min\{w(E),h(E)\}<w_{\min}$}
		\State $S\gets S+\mathrm{Area}(E)\,\hat p(E)$ \Comment{reuse the $N_{\min}$ samples}
		\State \Return
		\EndIf
		\For{each of the four children $E_c$ of $E$}
		\State \Call{Budget}{$E_c,d+1$}
		\EndFor
		\EndFunction
		\State \Call{Budget}{$E_0,0$}
		\State \Return $S$
	\end{algorithmic}
\end{algorithm}

\begin{theorem}[Computational complexity]
	\label{thm:complexity}
	Let the quadtree be built to depth $D=\lceil\log_{2}(L/\delta)\rceil$ with $\delta=\varepsilon_{\mathrm{rel}}L$. The deterministic variant creates $N_{\mathrm{cell}}=O(\mathrm{Perim}(\Gamma)/\delta)$ cells (Remark~\ref{rem:depth}), classifies each in $O(m+n)$ time, and evaluates each leaf in $O(1)$ time, or $O(s^{2})$ with the $s\times s$ midpoint rule. Hence it runs in
	\begin{equation}
		O\!\left((m+n)\,\frac{\mathrm{Perim}(\Gamma)}{\delta}\right)
		=O\!\left(\frac{m+n}{\varepsilon_{\mathrm{rel}}}\right)
		\label{eq:time}
	\end{equation}
time and $O(\mathrm{Perim}(\Gamma)/\delta)=O(1/\varepsilon_{\mathrm{rel}})$ memory. Since the relative error is $O(\delta/L)=O(\varepsilon_{\mathrm{rel}})$ (Proposition~\ref{prop:approx}), the cost is linear in the inverse relative error. The error-budget variant visits $O(N_{\mathrm{leaf}})$ cells, each sampled with $N_{\min}$ points whose membership tests cost $O(m+n)$ per point, so it runs in $O\bigl(N_{\mathrm{leaf}}N_{\min}(m+n)\bigr)$ time with $O(N_{\mathrm{leaf}})$ memory.
\end{theorem}

\begin{proof}
	By Remark~\ref{rem:depth} the adaptive tree contains $O(\mathrm{Perim}(\Gamma)/\delta)$ cells rather than $O(4^{D})$. Testing one cell against $P$ costs $O(m)$ with explicit ray casting and testing it against the $n$ circles costs $O(n)$, so the classification costs $O\bigl((m+n)\mathrm{Perim}(\Gamma)/\delta\bigr)$. The deterministic evaluation adds $O(\mathrm{Perim}(\Gamma)/\delta)$ operations (or $s^{2}$ point tests per leaf). Choosing $\delta=\varepsilon_{\mathrm{rel}}L$ and using $\mathrm{Perim}(\Gamma)=O(L)$ gives $O\bigl((m+n)/\varepsilon_{\mathrm{rel}}\bigr)$. For the error-budget variant every visited cell is sampled once with $N_{\min}$ points, and each point requires $O(m)$ for the polygon test and $O(n)$ for the circles. Since internal nodes are $O(N_{\mathrm{leaf}})$, the total work is $O\bigl(N_{\mathrm{leaf}}N_{\min}(m+n)\bigr)$.
\end{proof}

\begin{proposition}[Per-cell sampling error]\label{prop:sampling}
Let $E$ be a BDY leaf evaluated by the error-budget rule with $N_{\min}$ i.i.d. uniform points, let $p(E)=\mathrm{Area}(E\cap\Omega)/\mathrm{Area}(E)$ and let $\hat p(E)$ be the sample mean. With probability approximately $1-\alpha$,
	\begin{equation}
		\bigl|\mathrm{Area}(E)\,\hat p(E)-\mathrm{Area}(E)\,p(E)\bigr|
		\le \mathrm{Area}(E)\,\Delta(E),
		\label{eq:sampling}
	\end{equation}
	where $\Delta(E)$ is the Wilson half-width Eq.\eqref{eq:wilson} with $Z=z_{1-\alpha/2}$ ($Z=1.96$ for $\alpha=0.05$).
\end{proposition}
\begin{proof}
	The Wilson score interval covers the binomial proportion $p(E)$ with probability approximately $1-\alpha$. The points of $E$ are i.i.d. and Bernoulli$(p(E))$ with respect to $\Omega$, so multiplying the interval for $p(E)$ by $\mathrm{Area}(E)$ yields Eq.\eqref{eq:sampling}.
\end{proof}

\begin{proposition}[Approximation error of the deterministic rule]
	\label{prop:approx}
	Let $\delta$ be the finest accepted cell side. Every accepted BDY leaf has side at most $\delta$, hence area at most $\delta^{2}$, and by Eq.~\eqref{eq:ncell} there are $N_{\mathrm{bdy}}=O(\mathrm{Perim}(\Gamma)/\delta)$ of them. For both $s=1$ and $s=8$,
	\begin{equation}
		\bigl|A_{\mathrm{approx}}-A\bigr|
		\le\sum_{E\in\mathrm{BDY}}\mathrm{Area}(E)
		\le N_{\mathrm{bdy}}\,\delta^{2}
		=O\bigl(\mathrm{Perim}(\Gamma)\,\delta\bigr).
		\label{eq:approx-error}
	\end{equation}
	If in addition $A=\Theta(L^{2})$ and $\mathrm{Perim}(\Gamma)=O(L)$, the relative error is $O(\delta/L)=O(\varepsilon_{\mathrm{rel}})$.
\end{proposition}
\begin{proof}
For every leaf the estimate is $\mathrm{Area}(E)\hat p(E)$ with $\hat p(E)\in[0,1]$, while the exact value is $\mathrm{Area}(E\cap\Omega)\in[0,\mathrm{Area}(E)]$. Therefore, the absolute error of one leaf never exceeds $\mathrm{Area}(E)$. INT and EXT leaves are exact, so the sum runs over BDY leaves only, and each has area at most $\delta^{2}$. Using $N_{\mathrm{bdy}}\le\mathrm{Perim}(\Gamma)/\delta+O(1)$ (Eq.~\eqref{eq:ncell}) gives Eq.~\eqref{eq:approx-error}. The bound is a worst-case estimate. In practice the errors of neighbouring leaves have mixed signs, and the observed error is two to three orders of magnitude below it.
\end{proof}

\begin{remark}
	In our implementation every BDY leaf is evaluated by the midpoint rule Eq.\eqref{eq:midpoint} or by the error-budget rule Eq.\eqref{eq:refine-bud}, and all statements of this section refer to these two rules. For the error-budget, variant each accepted leaf satisfies $\mathrm{Area}(E)\Delta(E)\le\tau A_{\mathrm{bbox}}$. Summing Proposition~\ref{prop:sampling} over the $N_{\mathrm{leaf}}$ accepted leaves and applying a union bound gives the conservative bound $N_{\mathrm{leaf}}\tau A_{\mathrm{bbox}}$, with probability at least $1-N_{\mathrm{leaf}}\alpha$. The observed error decreases sublinearly in $\tau$ (Section~4), so no linear-in-$\tau$ bound is claimed.
\end{remark}

\section{Numerical experiments}

To test the performance of the proposed algorithm, we adopt relative error and runtime as evaluation metrics. We compare our algorithm with five classical estimators: Monte Carlo (MC), uniform grid (UG), adaptive subdivision (AS), grid integration (GI) and triangulation (Tri.). All experiments were conducted
on a workstation equipped with an Intel Core Ultra 9 275HX processor (24 cores, 24 threads) and 32 GB of DDR5 RAM. The reported runtime is single-threaded unless otherwise noted. The algorithm was implemented in Python 3.12.5 with NumPy 2.3.5 and Shapely 2.1.2 for geometric operations.

In addition, we use boundary integration (BI) as the ground truth. BI converts the area integral into a line integral along $\partial\Omega$ and computes it to machine precision for the polygonal circle representation used throughout. The resulting area $A_{\mathrm{BI}}$ is used to compute the relative error of every other method:
\begin{equation}
	\varepsilon_{\mathrm{rel}}=\frac{|\hat A_{\mathrm{method}}-A_{\mathrm{BI}}|}{A_{\mathrm{BI}}}.
\end{equation}

\subsection{Synthetic experiments with controlled parameters}

\subsubsection{Experimental design}
First, we construct two synthetic test cases with fully controlled parameters. In all experiments the polygonal domain $P$ is a \emph{star-shaped} and hence non-self-intersecting polygon. Its vertices are generated by drawing $n$ angles uniformly on $[0,2\pi]$, sorting them, and assigning each vertex a radius drawn uniformly from $[3,5]$. The resulting point set is then uniformly scaled so that the diameter of the polygon equals $diam(P)=10$ exactly. The circle centers are uniformly distributed within the
square $[-4,4]\times[-4,4]$, and the circle radii are uniformly drawn from $[1.0,2.5]$. This choice of parameters ensures substantial overlap between the circles and the polygonal region, producing complex intersection boundaries that rigorously challenge the proposed algorithm. To guarantee
reproducibility, every random generation is performed with the fixed random seed $42$.

\textbf{Case 1 (varying polygon complexity).}
The number of polygon vertices varies over the selected sequence $3,5,8,10,15,20,25,30,40,50$, while the number of circles is fixed at $30$. This case evaluates the sensitivity of each method to the geometric complexity of the polygonal domain.

\textbf{Case 2 (varying circle density).}
The number of circles ranges over the selected sequence $1,2,5,10,15,20,25,30$, while the polygon vertex count is fixed at $50$. This case assesses the algorithmic performance as the number of overlapping circles increases.

For a fair comparison, all seven methods use identical parameter settings to those in the real-instance comparison (Section 4.2), namely: Monte Carlo (MC) with $4{,}000{,}000$ samples and random seed $0$; Uniform Grid (UG) with $n=200$ cells per side; Adaptive Subdivision (AS) with maximum depth $8$; Grid Integration (GI) with $401$ slices; Triangulation (Tri.) with $2{,}000$ samples per triangle and seed $42$; Boundary Integration (BI), which computes the exact intersection area by Boolean operations and serves as the reference; and the proposed Adaptive Quadtree (AQ) with maximum depth $13$ and minimum cell side $w_{min}=0.002$. For every instance the reference area is taken from BI, and the accuracy of each approximate method is reported as the relative error $\varepsilon_{\mathrm{rel}}$. Each method is timed by taking the minimum wall-clock time over three independent repetitions, in order to reduce the influence of system jitter. The full numerical results are listed in Tables~\ref{tab:synth1} and~\ref{tab:synth2}, and the relative error and time scaling curves are shown in Fig.~\ref{Fig:7} and~\ref{Fig:6}.

\begin{figure}[H]
	\centering
	\includegraphics[width=0.7\textwidth]{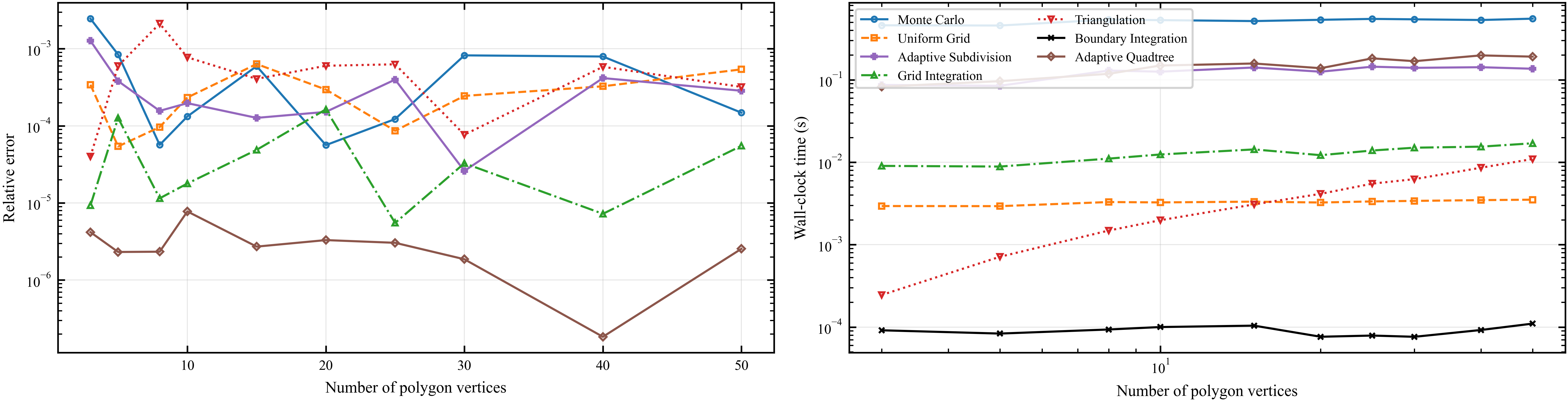}
	\caption{Relative error (left, log scale) and wall-clock time in seconds (right, log-log scale) of the seven methods on synthetic Case 1, where the number of polygon vertices varies from $3$ to $50$ while the number of circles is fixed at $30$. Boundary Integration is the exact reference (zero relative error) and is therefore omitted from the left panel. Its runtime is shown in the right panel.}
	\label{Fig:7}
\end{figure}

\begin{figure}[H]
	\centering
	\includegraphics[width=0.7\textwidth]{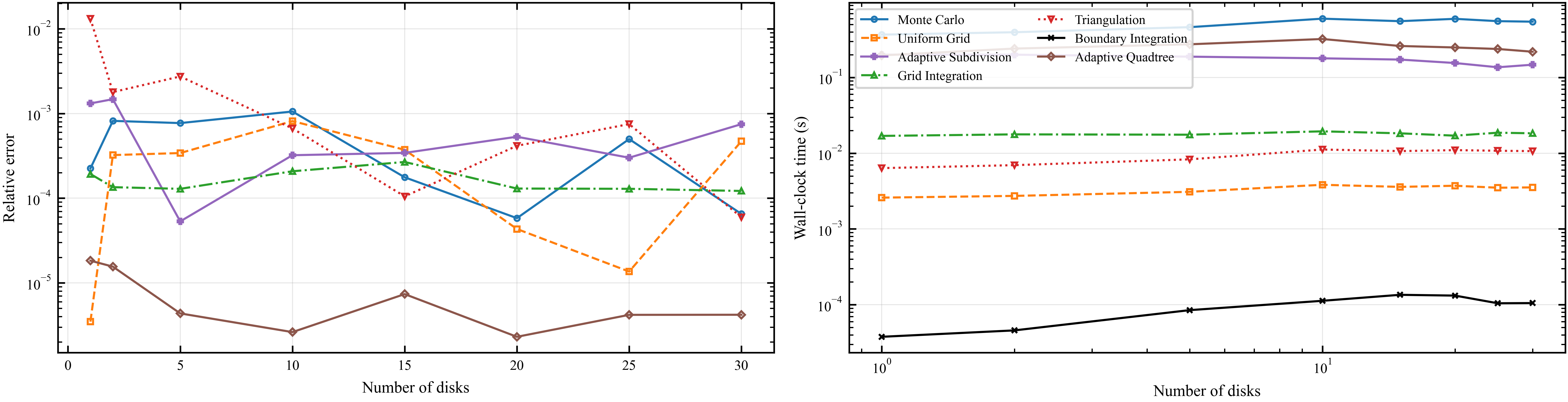}
	\caption{Relative error (left, log scale) and wall-clock time in seconds (right, log-log scale) of the seven methods on synthetic Case 2, where the number of circles ranges from $1$ to $30$ while the polygon has $50$ vertices. Boundary Integration is the exact reference (zero relative error) and is omitted from the left panel. Its runtime is shown in the right panel.}
	\label{Fig:6}
\end{figure}

\subsubsection{Significance test}

To assess whether the accuracy advantage of the proposed adaptive quadtree (AQ) is statistically robust, we treat each synthetic case as a repeated-measures design. Every parameter value (polygon vertex count in Case 1, circle count in Case~2) defines one geometric condition under which all six approximate methods are evaluated on the same instance (seed $42$). This yields $N=10$ paired observations per method in Case 1 and $N=8$ in Case 2. As the per-condition relative errors are not normally distributed and are heavily right-skewed, particularly for Monte Carlo and Triangulation on small-area intersections, we rely on non-parametric tests. The \emph{Friedman omnibus test} first verifies that the six methods differ overall. We then perform \emph{paired Wilcoxon signed-rank tests} for AQ against each competitor, with \emph{Holm step-down correction} for the five pairwise comparisons. Effect size is quantified via \emph{Cliff's $\delta$}, the fraction of pairwise cross-instance comparisons where the competitor's error exceeds AQ's error minus the fraction where it is smaller. Positive values mean AQ is more accurate, and $\delta=+1.0$ indicates a strict win on every comparison.

Table~\ref{tab:synth1_stats} reports the results for Case~1. The Friedman omnibus test is highly significant ($\chi^2=32.34$, $p=5.1\times 10^{-6}$), confirming that the six methods differ overall. Strikingly, AQ attains a \emph{10/10 win--loss record} against every competitor. It is the most accurate method on all ten geometric conditions. Cliff's $\delta$ ranges from $+0.96$ (against Grid Integration, the second-best method) to $+1.00$, and all Holm-corrected $p$-values are $0.0098<0.01$. By conventional benchmarks for Cliff's $\delta$, values above 0.474 constitute a "large" effect. Here $\delta\approx +1.0$ is an extreme effect size, far exceeding what is typically seen in computational experiments.

Table~\ref{tab:synth2_stats} reports the results for Case 2. The Friedman omnibus test is again significant ($\chi^2=19.29$, $p=1.7\times 10^{-3}$). AQ attains an \emph{8/8 win--loss record} against Monte Carlo, Adaptive Subdivision, Grid Integration, and Triangulation ($\delta=+1.00$, all Holm-corrected $p=0.0391<0.05$), and a 7/1 record against Uniform Grid ($\delta=+0.75$, Holm-corrected $p=0.0391<0.05$). The single loss against Uniform Grid occurs in the simplest possible configuration: a single circle intersecting a small star-shaped polygon, where the boundary consists only of polygon edges and arcs of a single circle, and the geometry is so regular that a coarse uniform grid (roughly $40{,}000$ evaluation points on a $200\times 200$ grid) already yields near-analytical accuracy. As soon as the circle count reaches two, the boundary becomes genuinely complex (multi-connected regions, overlapping arcs), and AQ regains a $20\times$ accuracy advantage over Uniform Grid. This edge-case behaviour is consistent with the design goal of AQ. It is tailored for complex intersection boundaries, and its superiority manifests precisely where other methods struggle.

In summary, over both synthetic cases the proposed adaptive quadtree is statistically and practically the most accurate approximate method, with a zero-failure record against four out of five competitors and a single well-understood loss against one of the fastest baseline in one degenerate configuration. The effect size ($\delta\approx +1.0$) is far larger than conventional thresholds, indicating that the accuracy gap is not a statistical artefact but a systematic property of the boundary-focused refinement strategy.

\begin{sidewaystable}[!ht]
	\centering
	\caption{Relative error and runtime (s) of the seven methods on synthetic Case 1 (30 circles, varying polygon complexity; Boundary Integration is the exact reference; cell entries: relative error / runtime).}
	\label{tab:synth1}
	\begin{tabular}{rccccccc}
		\hline
		Param & MC & UG & AS & GI & Tri. & BI & AQ \\
		\hline
		3 & 2.5e-03/0.46 & 3.4e-04/0.00 & 1.3e-03/0.09 & 9.3e-06/0.01 & 4.0e-05/0.00 & 0.0e+00/0.00 & 4.2e-06/0.08 \\
		5 & 8.4e-04/0.46 & 5.4e-05/0.00 & 3.8e-04/0.08 & 1.3e-04/0.01 & 6.0e-04/0.00 & 0.0e+00/0.00 & 2.3e-06/0.10 \\
		8 & 5.7e-05/0.54 & 9.6e-05/0.00 & 1.6e-04/0.13 & 1.1e-05/0.01 & 2.1e-03/0.00 & 0.0e+00/0.00 & 2.3e-06/0.12 \\
		10 & 1.3e-04/0.53 & 2.3e-04/0.00 & 2.0e-04/0.13 & 1.8e-05/0.01 & 7.7e-04/0.00 & 0.0e+00/0.00 & 7.8e-06/0.15 \\
		15 & 6.0e-04/0.52 & 6.3e-04/0.00 & 1.3e-04/0.14 & 4.9e-05/0.01 & 4.1e-04/0.00 & 0.0e+00/0.00 & 2.7e-06/0.16 \\
		20 & 5.6e-05/0.53 & 3.0e-04/0.00 & 1.5e-04/0.13 & 1.6e-04/0.01 & 6.0e-04/0.00 & 0.0e+00/0.00 & 3.3e-06/0.14 \\
		25 & 1.2e-04/0.55 & 8.6e-05/0.00 & 4.0e-04/0.14 & 5.5e-06/0.01 & 6.3e-04/0.01 & 0.0e+00/0.00 & 3.0e-06/0.18 \\
		30 & 8.2e-04/0.54 & 2.4e-04/0.00 & 2.6e-05/0.14 & 3.3e-05/0.02 & 7.7e-05/0.01 & 0.0e+00/0.00 & 1.9e-06/0.17 \\
		40 & 7.9e-04/0.53 & 3.3e-04/0.00 & 4.2e-04/0.14 & 7.2e-06/0.02 & 5.8e-04/0.01 & 0.0e+00/0.00 & 1.8e-07/0.20 \\
		50 & 1.5e-04/0.55 & 5.4e-04/0.00 & 2.9e-04/0.14 & 5.5e-05/0.02 & 3.2e-04/0.01 & 0.0e+00/0.00 & 2.5e-06/0.19 \\
		\hline
	\end{tabular}
	
	\vspace{2em} 
	\centering
	\caption{Relative error and runtime (s) on synthetic Case 2 (50-vertex polygon, varying circle density).}
	\label{tab:synth2}
	\begin{tabular}{rccccccc}
		\hline
		Param & MC & UG & AS & GI & Tri. & BI & AQ \\
		\hline
		1 & 2.2e-04/0.37 & 3.5e-06/0.00 & 1.3e-03/0.20 & 1.9e-04/0.02 & 1.3e-02/0.01 & 0.0e+00/0.00 & 1.8e-05/0.20 \\
		2 & 8.2e-04/0.40 & 3.2e-04/0.00 & 1.5e-03/0.20 & 1.3e-04/0.02 & 1.8e-03/0.01 & 0.0e+00/0.00 & 1.6e-05/0.24 \\
		5 & 7.7e-04/0.46 & 3.4e-04/0.00 & 5.3e-05/0.19 & 1.3e-04/0.02 & 2.7e-03/0.01 & 0.0e+00/0.00 & 4.4e-06/0.27 \\
		10 & 1.1e-03/0.60 & 8.1e-04/0.00 & 3.2e-04/0.18 & 2.1e-04/0.02 & 6.7e-04/0.01 & 0.0e+00/0.00 & 2.6e-06/0.32 \\
		15 & 1.8e-04/0.55 & 3.7e-04/0.00 & 3.4e-04/0.17 & 2.7e-04/0.02 & 1.1e-04/0.01 & 0.0e+00/0.00 & 7.4e-06/0.26 \\
		20 & 5.8e-05/0.59 & 4.3e-05/0.00 & 5.3e-04/0.16 & 1.3e-04/0.02 & 4.2e-04/0.01 & 0.0e+00/0.00 & 2.3e-06/0.25 \\
		25 & 5.0e-04/0.55 & 1.4e-05/0.00 & 3.0e-04/0.14 & 1.3e-04/0.02 & 7.5e-04/0.01 & 0.0e+00/0.00 & 4.2e-06/0.24 \\
		30 & 6.5e-05/0.55 & 4.7e-04/0.00 & 7.5e-04/0.15 & 1.2e-04/0.02 & 5.9e-05/0.01 & 0.0e+00/0.00 & 4.2e-06/0.22 \\
		\hline
	\end{tabular}
\end{sidewaystable}

\begin{table}[htbp]
	\centering
	\caption{\textbf{Case 1: varying polygon complexity (30 circles).}
		Wilcoxon signed-rank test and Cliff's $\delta$ for the relative error of the proposed adaptive quadtree (AQ) versus each competitor, across $N=10$ geometric conditions (polygon vertex count from $3$ to $50$). \textit{Wins/Losses} counts how many conditions AQ outperforms/is outperformed by the competitor. A \emph{positive} Cliff's $\delta$ means AQ is more accurate than the competitor overall. $^{*}$$p<0.05$, $^{**}$$p<0.01$, $^{***}$$p<0.001$. Omnibus Friedman test over the six non-exact methods:
		$\chi^2=32.34$, $p=5.08e-06$.}
	\label{tab:synth1_stats}
	\begin{tabular}{lcccccc}
		\hline
		Competitor & Wins / Losses & Median cmp & Median AQ & Cliff's $\delta$ & Holm $p$-value \\
		\hline
MC  & 10/0 & $\num{3.7e-04}$ & $\num{2.6e-06}$ & $+1.000$ & 0.0098 $^{**}$ \\
UG  & 10/0 & $\num{2.7e-04}$ & $\num{2.6e-06}$ & $+1.000$ & 0.0098 $^{**}$ \\
AS  & 10/0 & $\num{2.4e-04}$ & $\num{2.6e-06}$ & $+1.000$ & 0.0098 $^{**}$ \\
GI  & 10/0 & $\num{2.5e-05}$ & $\num{2.6e-06}$ & $+0.960$ & 0.0098 $^{**}$ \\
Tri & 10/0 & $\num{5.9e-04}$ & $\num{2.6e-06}$ & $+1.000$ & 0.0098 $^{**}$ \\
		\hline
	\end{tabular}
\end{table}

\begin{table}[htbp]
	\centering
	\caption{\textbf{Case 2: varying circle density (50-vertex polygon).}
Wilcoxon signed-rank test and Cliff's $\delta$ for the relative error of the proposed adaptive quadtree (AQ) versus each competitor, across $N=8$ geometric conditions (circle count from $1$ to $30$). \textit{Wins/Losses} counts how many conditions AQ outperforms/is outperformed by the competitor. A \emph{positive} Cliff's $\delta$
means AQ is more accurate than the competitor overall.	$^{*}$$p<0.05$, $^{**}$$p<0.01$, $^{***}$$p<0.001$. Omnibus Friedman test over the six non-exact methods: $\chi^2=19.29$, $p=1.70e-03$.}
	\label{tab:synth2_stats}
	\begin{tabular}{lcccccc}
		\hline
Competitor & Wins / Losses & Median cmp & Median AQ & Cliff's $\delta$ & Holm $p$-value \\
		\hline
MC  & 8/0 & $\num{3.6e-04}$ & $\num{4.3e-06}$ & $+1.000$ & 0.0391 $^{*}$ \\
UG  & 7/1 & $\num{3.3e-04}$ & $\num{4.3e-06}$ & $+0.750$ & 0.0391 $^{*}$ \\
AS  & 8/0 & $\num{4.4e-04}$ & $\num{4.3e-06}$ & $+1.000$ & 0.0391 $^{*}$ \\
GI  & 8/0 & $\num{1.3e-04}$ & $\num{4.3e-06}$ & $+1.000$ & 0.0391 $^{*}$ \\
Tri & 8/0 & $\num{7.1e-04}$ & $\num{4.3e-06}$ & $+1.000$ & 0.0391 $^{*}$ \\
		\hline
	\end{tabular}
\end{table}

\subsubsection{Results and discussion on synthetic instances}

Tables~\ref{tab:synth1} and~\ref{tab:synth2} list the relative error and runtime of all seven methods on the two synthetic cases, and Figs.~\ref{Fig:7} and~\ref{Fig:6} visualize the error and wall-clock time as functions of the polygon vertex count and the number of circles, respectively. Boundary Integration (BI) is the exact reference and therefore always attains zero error. It is omitted from the log-scale error panels but included in the timing panels.

\textbf{Accuracy.}
In Case 1 the proposed adaptive quadtree (AQ) attains the smallest relative error on all ten geometric conditions, and in Case 2 on all conditions except the single-circle instance discussed below. More importantly, its error stays at the $10^{-6}$ level in the median and varies by less than about two orders of magnitude over the whole range of parameters.  In Case 1 (varying polygon complexity, $30$ circles) the median relative error of AQ is $2.6\times 10^{-6}$ with a maximum of $7.8\times 10^{-6}$ over all vertex counts from $3$ to $50$. It even reaches $1.8\times 10^{-7}$ at $40$ vertices. This is roughly one order of magnitude lower than the second-best method, grid integration (GI, median $2.5\times 10^{-5}$), and about two
orders of magnitude lower than the remaining methods (uniform grid, adaptive subdivision, Monte Carlo, and triangulation), whose median errors all lie in the $10^{-4}$ range. In Case 2 (varying circle density, $50$-vertex polygon), the same pattern holds: AQ keeps a median relative error of $4.3\times 10^{-6}$, again about one to two orders of magnitude below GI ($1.3\times 10^{-4}$) and the other approximate methods ($10^{-4}$--$10^{-3}$).

The synthetic instances also reveal the weaknesses of the competing methods in a controlled way. Monte Carlo and triangulation are particularly fragile when the intersection area is small. In Case~1 the smallest polygon ($3$ vertices, reference area only $5.83$) yields a Monte Carlo error of $2.5\times 10^{-3}$, because very few samples fall inside the tiny intersection and the relative variance becomes large. In Case~2 the single-circle instance (reference area $3.71$) pushes the triangulation error to $1.3\times 10^{-2}$, since only a handful of triangles overlap the circle and the per-triangle sampling variance dominates. By contrast, AQ maintains a median relative error at the $10^{-6}$ level, with individual values ranging from $1.8\times10^{-7}$ to $1.8\times10^{-5}$ throughout these extreme configurations. This stability directly reflects the boundary-focused design: refinement is concentrated exactly where the intersection boundary lies, so neither an increase in polygon complexity nor an increase in circle density introduces additional error in regions away from the boundary.

\textbf{Runtime.}
The cheapest methods are uniform grid and triangulation (both at or below about $0.01$ s on these synthetic instances), followed by grid integration ($\approx 0.01$--$0.02$ s). The proposed adaptive quadtree and adaptive subdivision occupy the same moderate tier, with medians of $0.155$ s and $0.135$ s in Case 1 and $0.245$ s and $0.175$ s in Case 2, so AS is marginally cheaper in our implementation (by a factor of about $1.15$--$1.4$). Monte Carlo is the most expensive, at $0.46$--$0.55$ s per instance in Case 1 and $0.37$--$0.60$ s in Case 2, dominated by its fixed $4{,}000{,}000$ samples, so AQ is about $2$--$3.5$ times faster than Monte Carlo. We emphasise that the absolute runtimes of the recursive methods (AS and AQ) depend on the implementation and on the minimum-cell-width safeguard, and we therefore do not rank them against each other by speed. The practically relevant observation is that AQ attains its $10^{-6}$-level accuracy, one to two orders of magnitude better than every competitor.

\textbf{Summary.}
Across both synthetic cases the adaptive quadtree is consistently the most accurate approximate method, with a median relative error at the $10^{-6}$ level (individual values between $1.8\times10^{-7}$ and $1.8\times10^{-5}$), far more stable than the competing methods, which degrade to the \(10^{-4}\)--\(10^{-2}\) range and suffer pronounced variance in small-area configurations. Combined with the real-instance comparison in Table~\ref{tab:4}, these controlled experiments confirm that the proposed method delivers both high accuracy and moderate cost over a wide range of geometric settings.

\subsection{Real-world scenario simulation}

To further evaluate the performance of the proposed algorithm, we conduct experiments on complex non-convex real-world scenarios. Land polygons are obtained from \href{https://www.naturalearthdata.com/downloads/10m-physical-vectors/}{\textit{Natural Earth}} (1:10m physical vectors) and clipped to the Caribbean window $[90^{\circ}\mathrm{W},57^{\circ}\mathrm{W}]\times[7^{\circ}\mathrm{N},28^{\circ}\mathrm{N}]$. Islets smaller than $0.004\ \mathrm{deg}^2$ are discarded, and the remaining polygons are simplified with Douglas-Peucker ($\varepsilon_{\mathrm{DP}} =0.0476^{\circ}$, topology preserved), yielding a 97-part multipolygon with 1{,}284 vertices and a land area of $133.138\ \mathrm{deg}^2$ inside the bounding box. The region features highly irregular coastlines with multiple concave bays (e.g., the Gulf of Honduras), peninsulas (e.g., Yucat\'an), and islands (Cuba, Hispaniola, Jamaica, Puerto Rico, Trinidad, and the Lesser Antilles). This geometric complexity makes it an ideal test case for evaluating the performance of coverage-area approximation algorithms.

To simulate realistic base-station signal coverage, we deploy 100 sensor nodes, each modeled as a circle. Circle centers are uniformly distributed in $[90^{\circ}\mathrm{W},57^{\circ}\mathrm{W}]\times[7^{\circ}\mathrm{N},28^{\circ}\mathrm{N}]$ and radii follow $U(100,280)$~km, converted to degrees by $r/(111.32\cos\varphi)$, using a fixed seed of 42. Their union covers $85.6891\ \mathrm{deg}^2$ of the $133.138\ \mathrm{deg}^2$ land area ($64.36\%$), which is taken as the reference for all estimator comparisons. The bounding-box area is $A_{\mathrm{bbox}}\approx 693\ \mathrm{deg}^2$. The internal sampling of each stochastic estimator also uses a fixed seed (Monte Carlo: seed 0 with $4\times10^{6}$ samples; triangulation: seed 42 with 2{,}000 samples per triangle), so all results are fully reproducible. The adaptive quadtree used in the method comparison (Table~\ref{tab:4}) is deterministic and therefore requires no seed. Fig.~\ref{Fig:9} visualizes the polygonal region and the distribution of the 100 sensor nodes. Table~\ref{tab:4} summarizes the performance of various area estimation algorithms against the reference coverage.

\begin{figure}[H]
	\centering
	\includegraphics[width=0.7\textwidth]{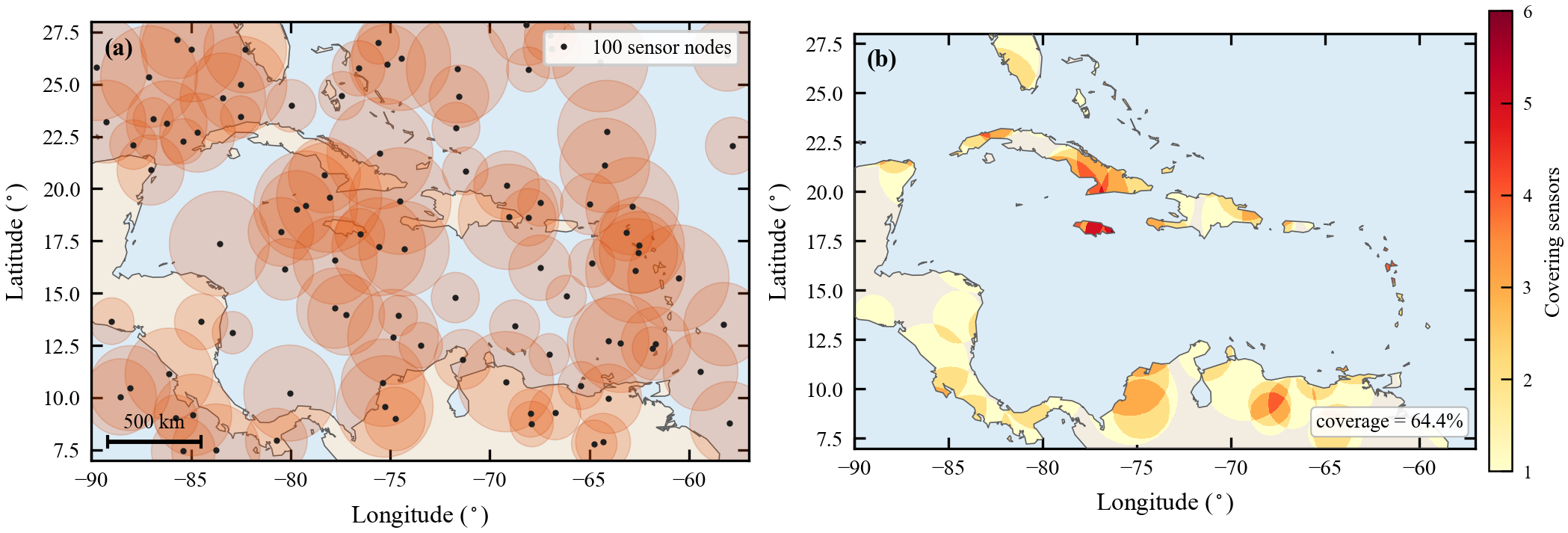}
	\caption{Signal coverage simulation of base stations established along the coastline of the Caribbean Sea.}
	\label{Fig:9}
\end{figure}

\begin{table}[htbp]
	\centering
	\caption{Final results of the seven coverage-area estimators on the real
		Caribbean coastline instance (1{,}284 vertices, 97 polygon parts, and 100
		sensor circles with seed~42). The land area is $133.138\ \mathrm{deg}^2$. The
		reference coverage area, obtained by exact boundary integration, is
		$85.6891\ \mathrm{deg}^2$, corresponding to a coverage ratio of $64.36\%$.}
	\label{tab:4}
	\begin{tabular}{lrrrr}
		\toprule
		\textbf{Method} & \textbf{Area (deg$^{2}$)} & \textbf{Abs Error} & \textbf{Rel Error} & \textbf{Time (s)} \\
		\midrule
		Monte Carlo & 85.8060 & 1.17e-01 & 0.136\% & 1.4547 \\
		Uniform Grid & 85.5855 & 1.04e-01 & 0.121\% & 0.0073 \\
		Adaptive Subdivision & 85.7156 & 2.65e-02 & 0.031\% & 0.1584 \\
		Grid Integration & 85.6848 & 4.31e-03 & 0.005\% & 0.2529 \\
		Triangulation & 85.4910 & 1.98e-01 & 0.231\% & 0.3586 \\
		Boundary Integration & 85.6891 & 0.00e+00 & 0.000\% & 0.0029 \\
		\textbf{Adaptive Quadtree (Ours)} & \textbf{85.6892} & \textbf{8.72e-05} & \textbf{0.0001\%} & 1.1538 \\
		\bottomrule
	\end{tabular}
\end{table}

\begin{figure}[H]
	\centering
	\includegraphics[width=0.7\textwidth]{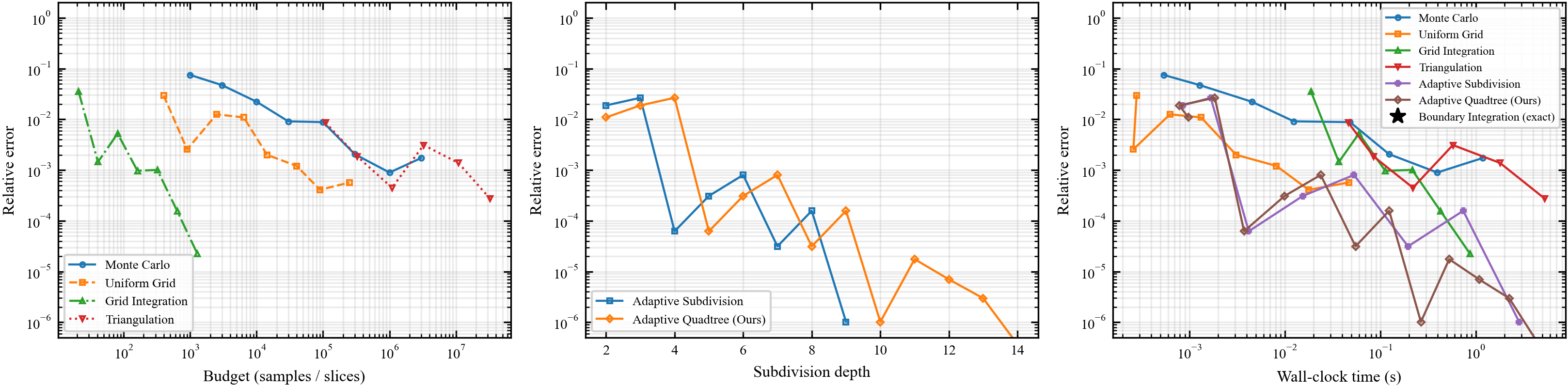}
	\caption{Convergence of the coverage-area estimators on the Caribbean coastline instance (1{,}284 vertices, 97 polygon parts, 100 sensor circles, seed 42). (a) Relative error versus sampling budget for the budget-controlled estimators (MC, UG, GI, and Tri.). (b) Relative error versus subdivision depth for the deterministic recursive methods (AS and AQ). (c) Relative error versus wall-clock time for all methods, each marked at a representative operating point; the star denotes the exact boundary integration reference.}
	\label{Fig:8}
\end{figure}

\subsubsection{Performance under Various Parameters}

Next, we examine how the accuracy--cost trade-off of the proposed quadtree is controlled by its parameters. All experiments use the same Caribbean instance as Table~\ref{tab:4} (1{,}284 coastline vertices, 97 land polygons, 100 sensor circles), and errors are measured against the boundary-integration reference of $85.6891$~deg$^2$. In contrast to the deterministic, depth-limited variant in Table~\ref{tab:4}, the variant studied here terminates refinement through a statistical error budget. At every level, each active cell $E$ is probed with $N_{\min}$ uniform random points, yielding a coverage fraction $\hat{p}(E)$. The half-width of the $95\%$ Wilson score interval, $\Delta(E)$, is given by Eq.~\eqref{eq:wilson} and is converted into an absolute area uncertainty $\mathrm{Area}(E)\,\Delta(E)$. A cell is accepted and contributes $\mathrm{Area}(E)\,\hat{p}(E)$ to the estimate as soon as $\mathrm{Area}(E)\,\Delta(E)\le\tau A_{\mathrm{bbox}}$; otherwise it is split into four children. A maximum depth of $18$ and a minimum admissible cell side of $10^{-4}$ degrees act only as safeguards against pathological refinement. The tolerance is normalized by the bounding-box area, $\mathrm{tol}=\tau A_{\mathrm{bbox}}$ with $A_{\mathrm{bbox}}=693$~deg$^2$, and $\tau=C\times10^{-6}$. The dimensionless constant $C$ therefore controls the refinement aggressiveness (for $0.1\le C\le 7$, $\mathrm{tol}$ ranges from $6.9\times10^{-5}$ to $4.9\times10^{-3}$~deg$^2$), while $N_{\min}$ controls the sampling noise that enters the stopping rule.

Three sweeps were performed. First, $C$ varies from $0.1$ to $7.0$ in steps of $0.1$ for $N_{\min}\in\{50,100,200,500,1000\}$ (Fig.~\ref{Fig:10}). Second, $N_{\min}$ varies from $50$ to $1000$ in steps of $50$ for $C\in\{0.1,0.5,1,2,4,7\}$ (Fig.~\ref{Fig:11}). Third, $\tau$ is swept over eleven logarithmically spaced values between $10^{-7}$ and $10^{-2}$, crossed with the same five values of $N_{\min}$ (Fig.~\ref{Fig:14}). After removing duplicates, 485 distinct $(\tau,N_{\min})$ configurations were evaluated. Each configuration was executed once with the sampling seed fixed to $0$. Consequently, small fluctuations superimposed on the systematic trends reflect sampling noise of a single realization rather than algorithmic instability.

\begin{figure}[htbp]
	\centering
	\includegraphics[width=0.7\textwidth]{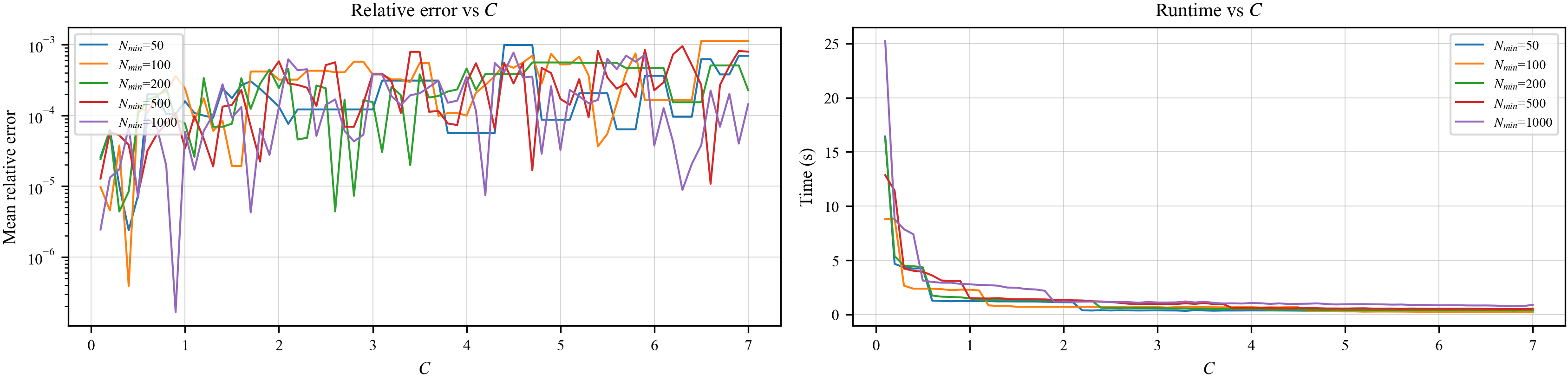}
	\caption{Sensitivity of the error-budget quadtree to the constant $C$, with $\tau=C\times10^{-6}$ and $\mathrm{tol}=\tau A_{\mathrm{bbox}}$. (a) Relative error and (b) wall-clock runtime for five per-cell sample sizes $N_{\min}$. The mean relative error grows from $9.2\times10^{-5}$ ($C<1$) to $4.5\times10^{-4}$ ($C\ge4$), whereas the runtime drops by a factor of about 39 as the error budget is relaxed.}
	\label{Fig:10}
\end{figure}

\begin{figure}[htbp]
	\centering
	\includegraphics[width=0.7\textwidth]{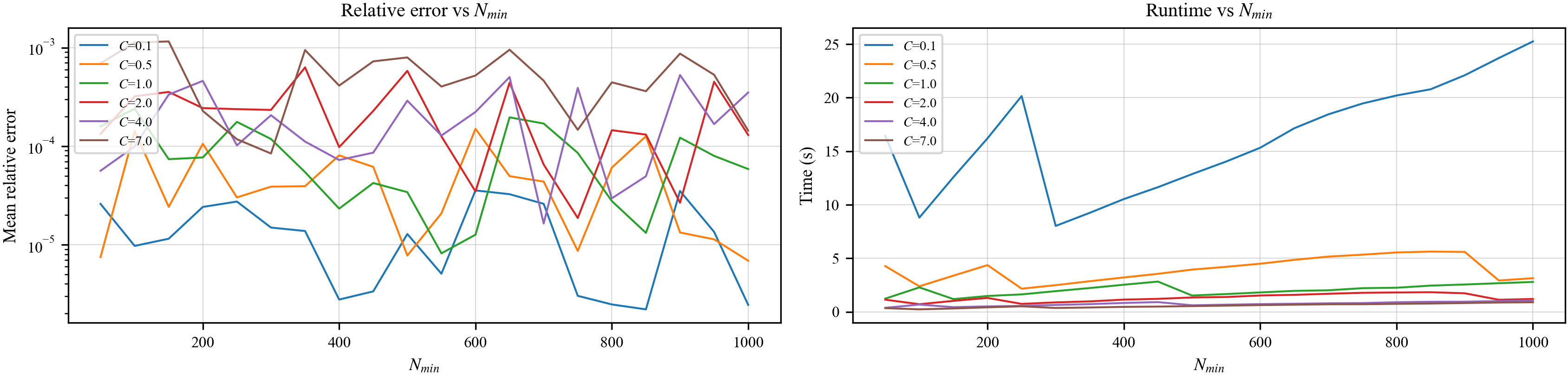}
	\caption{Sensitivity to the per-cell sample size $N_{\min}$. (a)~Relative error and (b)~wall-clock runtime for six values of $C$. The accuracy is relatively insensitive to $N_{\min}$; the runtime grows markedly with $N_{\min}$ only under the tightest budget ($C=0.1$) and is nearly constant otherwise.}
	\label{Fig:11}
\end{figure}

\begin{figure}[htbp]
	\centering
	\includegraphics[width=0.7\textwidth]{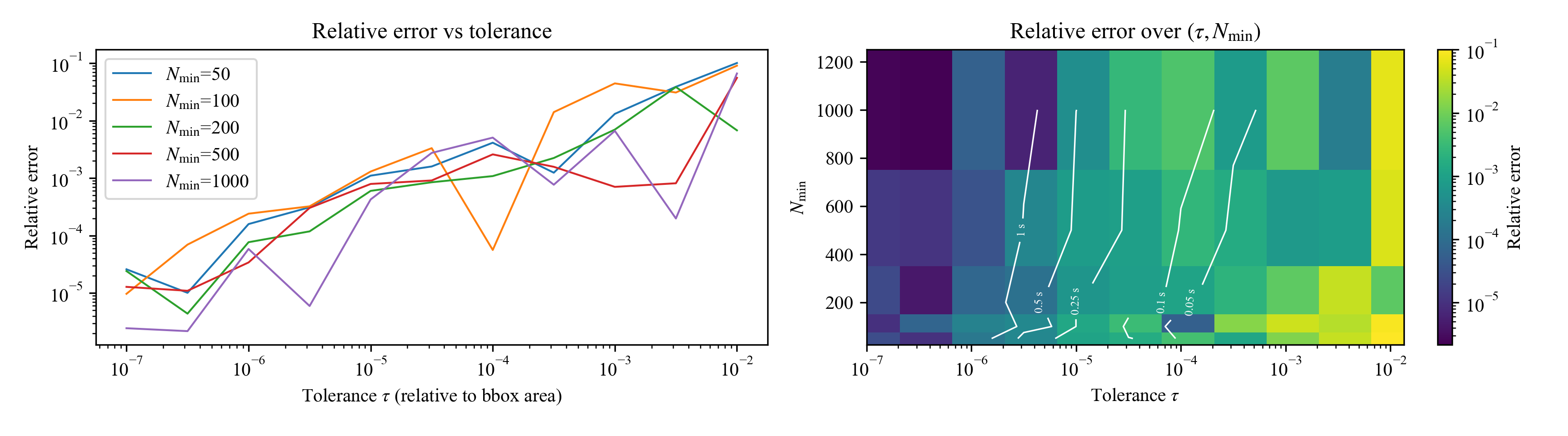}
\caption{Joint effect of the normalized tolerance $\tau$ and the per-cell sample size $N_{\min}$. (a) Relative error versus $\tau$ for five values of $N_{\min}$; (b) relative-error heat map over the $(\tau,N_{\min})$ plane, with white contours denoting iso-runtime levels in seconds. The error remains small for tight tolerances and grows sublinearly as $\tau\gtrsim10^{-5}$, with the largest errors occurring at the loosest tolerances.}
	\label{Fig:14}
\end{figure}

\subsubsection{Discussion}

\textbf{Algorithm performance.}
Based on the sensor distribution shown in Fig.~\ref{Fig:9}, we compute the effective coverage area using all seven methods. Boundary integration (BI) yields $85.6891\ \mathrm{deg}^2$ for the polygonal coastline with circular coverage, which
we take as the ground truth. The Monte Carlo (MC), uniform grid (UG), grid integration (GI), and triangulation (Tri.) estimators are all controlled by a sampling budget, and Fig.~\ref{Fig:8}(a) shows how their relative error decreases as this budget grows (random samples for MC, grid cells for UG, latitudinal slices for GI, and samples per triangle for Tri.). The adaptive subdivision (AS) algorithm and the proposed adaptive quadtree (AQ) have no such sampling parameter, because they recursively refine the domain until a depth threshold is reached. Their accuracy is instead reported as a function of subdivision depth in Fig.~\ref{Fig:8}(b). Both converge to the $10^{-6}$ level, with AQ reaching $3.7\times10^{-7}$ at depth 14. The small
non-monotonicities at intermediate depths stem from cell-boundary quantization and vanish at finer resolutions. To compare all methods on a common axis, Fig.~\ref{Fig:8}(c) plots the relative error against wall-clock time, with each method marked at a representative operating point. Table~\ref{tab:4}
summarizes the final coverage area, absolute error, relative error, and runtime. AQ attains the highest accuracy among all approximate methods, with a relative error of $1.0\times10^{-6}$ ($0.0001\%$), i.e. an absolute error of $8.7\times10^{-5}$~deg$^2$, at $1.1538$ s. This is one to three orders of magnitude lower than the budget-controlled and subdivision baselines.

\textbf{Convergence behaviour.}
Fig.~\ref{Fig:8}(a) shows how the relative error of the four
budget-controlled estimators evolves with their sampling budget. Monte Carlo (MC) and triangulation (Tri.) are randomized, whereas the uniform grid (UG) and grid integration (GI) are deterministic quadratures whose budget is the number of cells or latitudinal slices, respectively. MC follows the expected $O(N^{-1/2})$ decay over the whole range, decreasing from $1.4\times10^{-1}$ at $10^{3}$ samples to about $1.5\times10^{-3}$ at $3\times10^{6}$. The residual scatter around this law reflects the finite number of independent realizations used in the average. UG and GI decay faster, consistent with their algebraic convergence rates, and GI reaches the lowest error among the budget-controlled methods ($2.3\times10^{-5}$ with $1{,}281$ slices). Tri.\ is highly non-monotone ($2.7\times10^{-4}$ at $3\times10^{4}$ samples per triangle), because a fixed per-triangle sample size assigns a disproportionate share of the variance to the many small triangles produced by the irregular coastline.

Figure~\ref{Fig:8}(b) reports the two deterministic recursive
methods, which have no sampling budget and are therefore characterized by their subdivision depth. Both converge to the $10^{-6}$ level: adaptive subdivision (AS) reaches $1.0\times10^{-6}$ at its finest depth 9 ($5.2$ s), and the adaptive quadtree (AQ) reaches $3.7\times10^{-7}$ at depth 14 ($6.8$ s). At depths 9 and 10 the two are comparable, and AQ becomes more accurate only at finer depths. The curves are not strictly monotone at intermediate depths (for instance, AQ dips to $1.0\times10^{-6}$ at depth~10 and rises to $1.8\times10^{-5}$ at depth~11). This is a cell--boundary quantization artifact rather than a loss of convergence, and the lower envelope decreases steadily. The methods should therefore be assessed by their envelope or by the error at the finest resolution, not by individual intermediate points.

Since the budgets in Fig.~\ref{Fig:8}(a) count different objects (random samples, grid cells, slices, and per-triangle samples), the panel illustrates the internal convergence of each estimator rather than a cross-method comparison. The latter is provided by Fig.~\ref{Fig:8}(c), which plots the relative error against wall-clock time with every method at a representative operating point. AQ dominates all approximate methods in the high-accuracy regime. At the relative error of $1.0\times10^{-6}$, its best configuration (the $8\times8$ midpoint rule at depth 10) requires $0.40$ s, whereas the default centre-point configuration of Table~\ref{tab:4} attains the same accuracy in $1.1538$ s. AS needs $5.2$ s, and GI needs $1.34$ s to reach only $2.3\times10^{-5}$. UG dominates MC entirely
($4.1\times10^{-4}$ in $0.029$ s versus $9.0\times10^{-4}$ in $0.68$ s). Triangulation attains moderate accuracy ($2.7\times10^{-4}$) but is the slowest budget-controlled method, requiring $9.2$ s at its finest settings. Boundary integration (BI) remains exact, but it is shown only as an accuracy reference. Its cost grows with the number of polygon vertices and circles, so it cannot replace an approximation scheme in large-scale settings.

In practice, depths 9--10 provide a relative error between $10^{-4}$ and $10^{-6}$ in less than $0.5$~s and constitute the operating point of choice, while depth~14 should be reserved for applications that require extreme accuracy. Beyond depth~11 the error decreases by roughly one order of magnitude per three additional levels, at approximately twice the runtime per level. The refinement depth is thus an effective but not free tuning knob.

\subsubsection{Parameter sensitivity analysis}

\textbf{Constant $C$ (Fig.~\ref{Fig:10}).}
Across all configurations the relative error stays below $7.4\times10^{-4}$, while the runtime ranges from about $7.6$~s at $C=0.1$ to $0.20$~s at $C=7$. Relaxing $C$ from $0.1$ to $7$ reduces the runtime by a factor of about $39$, and most of this gain is already obtained at $C\approx1$--$2$, where all runs finish within a couple of seconds. The accuracy cost is markedly sublinear: the mean relative error grows from $9.2\times10^{-5}$ for $C<1$ to $1.9\times10^{-4}$ ($1\le C<2$), $3.5\times10^{-4}$ ($2\le C<4$) and $4.5\times10^{-4}$ ($C\ge4$). That is, a 70-fold relaxation of the per-cell budget increases the error only about fivefold. This is consistent with the budget being a per-cell 95\% bound. The global error aggregates many cell contributions with partially cancelling signs. The visible oscillations reflect the discrete set of accepted cells changing between configurations of a single realization.

\textbf{Samples per cell $N_{\min}$ (Fig.~\ref{Fig:11}).}
At fixed $C$ the error shows no monotone dependence on $N_{\min}$. The error level is governed mainly by $C$, rising from about $1.5\times10^{-5}$ at $C=0.1$ to $3.0\times10^{-4}$ at $C=4$. The runtime, however, is non-monotone under tight
budgets. At $C=0.1$ it equals $16.4$~s at $N_{\min}=50$, attains a minimum of $8.8$~s at $N_{\min}=100$, and then grows to $25.2$~s at $N_{\min}=1000$. For $C\ge1$ the runtime never exceeds $2.9$~s and depends only weakly on $N_{\min}$. This behaviour follows from the stopping rule: since the Wilson half-width scales as $N_{\min}^{-1/2}$, a small $N_{\min}$ leaves a large area uncertainty and forces deep refinement with many cheap cells, whereas a large $N_{\min}$ accepts coarse cells early but spends many samples on each. Spatial resolution and statistical precision therefore compensate each other, which explains both the flat error curves and the interior runtime
minimum.

\textbf{Tolerance $\tau$ (Fig.~\ref{Fig:14}).}
Three regimes are visible. For $\tau\le10^{-6}$ the estimate is refinement-limited. The error lies between $2.2\times10^{-6}$ and $2.4\times10^{-4}$, while the cost rises from $1.2$--$2.8$~s at $\tau=10^{-6}$ to $8.8$--$25.2$~s at $\tau=10^{-7}$. In the transition band $10^{-5}$--$10^{-4}$, the runtime falls to $0.03$--$0.50$~s while the error remains of order $10^{-4}$ to $5\times10^{-3}$. For $\tau\ge10^{-3}$ the budget is so loose that only one or two levels are generated: the runtime drops to $6$--$27$~ms at $\tau=10^{-3}$ and to a few milliseconds at $\tau=10^{-2}$, while the error grows to $7.0\times10^{-4}$--$4.4\times10^{-2}$ and $6.8\times10^{-3}$--$1.0\times10^{-1}$, respectively. Overall the error degrades sublinearly with $\tau$ once $\tau\gtrsim10^{-5}$. Tightening $\tau$ from $10^{-6}$ to $10^{-7}$ still reduces the error markedly, but at roughly a tenfold cost in runtime.

\textbf{Practical choice.}
$C\in[1,3]$ with $N_{\min}\in\{100,200\}$ offers the best compromise, yielding relative errors of about $8\times10^{-5}$ to $3.2\times10^{-4}$ in $0.6$--$2.3$ s. When a much tighter accuracy is required, the deterministic depth-limited variant is preferable to further reducing $\tau$. Its best configuration (the $8\times8$ midpoint rule at depth~10) attains a relative error of approximately $10^{-6}$ in $0.40$ s (Fig.~\ref{Fig:8}(c)), whereas the default centre-point configuration of Table~\ref{tab:4} reaches the same accuracy in $1.1538$ s and the tightest budgeted runs need tens of seconds.

\section{Conclusion and future work}

This paper presented a boundary-focused adaptive quadtree (AQ) algorithm for approximating the area of the intersection between a union of circles and a complex polygonal domain. By concentrating refinement exactly on the intersection boundary, the method keeps the approximation error at the $10^{-6}$ level for realistic instances, with only weak sensitivity to the polygon complexity and the circle density.

Two controlled synthetic studies, in which the polygon vertex count and the circle count were varied independently, showed that AQ attains a median relative error at the $10^{-6}$ level, with individual values between $1.8\times10^{-7}$ and $1.8\times10^{-5}$, one to two orders of magnitude lower than other estimators. A non-parametric repeated-measures analysis (the Friedman omnibus test, followed by Holm-corrected paired Wilcoxon signed-rank tests and Cliff's $\delta$) confirmed that this advantage is statistically robust. Over both cases AQ achieved a zero-failure record against four of five competitors, with the single loss occurring in the most regular configuration (a single circle), where a coarse uniform grid is already near-analytical. The effect size ($\delta \approx +1.0$) far exceeds conventional thresholds, indicating that the accuracy gap is a systematic property of boundary-focused refinement rather than sampling noise. On the real Caribbean coastline instance ($1{,}284$ vertices, $97$ polygon parts, $100$ circles), AQ reached a relative error of $1.0\times10^{-6}$ ($0.0001\%$), i.e. an absolute error of $8.7\times10^{-5}$ deg$^2$, in $1.15$ s, again the most accurate among all approximate methods.

A parameter sensitivity study of the error-budget variant showed that the accuracy is governed primarily by the normalized tolerance and is nearly insensitive to the per-cell sample size $N_{\min}$: a $70$-fold relaxation of the per-cell budget increases the error only about fivefold, while spatial resolution and statistical precision compensate each other. In practice, $C\in[1,3]$ with $N_{\min}\in\{100,200\}$ offers the best accuracy-cost compromise, whereas the deterministic depth-limited variant attains $10^{-6}$ in roughly $0.4$ s with the $8\times8$ midpoint rule ($1.1538$ s with the centre-point rule of Table~\ref{tab:4}) and is preferable when extreme accuracy is required.

The method is not without limitations. On highly regular intersections dominated by a single circle, a coarse uniform grid can match it, and the absolute runtime of the recursive scheme depends on the implementation and the minimum-cell-width safeguard. Natural directions for future work include anisotropic refinement and time-varying circle configurations, as well as an extension to three-dimensional problems such as the intersection of spheres with a polyhedron, which would broaden the applicability of the approach to volumetric coverage problems. In addition, since the Monte Carlo subsampling performed within each boundary cell is inherently parallel, a GPU-based implementation is a promising route toward real-time or large-scale deployment of the error-budget variant.

\noindent {\bf Acknowledgments.} The authors would like to thank  editors and reviewers sincerely for their constructive comments, which greatly improved our manuscript. 

\clearpage
\bibliographystyle{elsarticle-num-names} 
\bibliography{references}  

\end{document}